\documentclass[11pt,reqno]{amsart}
\pdfoutput=1
\usepackage[utf8]{inputenc}
\usepackage[
 a4paper,
 total={160mm,257mm},
 left=25mm,
 top=20mm]{geometry}
\usepackage{hyperref}
\usepackage{amsfonts}
\usepackage{amsmath}
\usepackage{amssymb}
\usepackage{mathrsfs}  
\usepackage{color}
\usepackage{physics}
\usepackage{multicol}
\usepackage{setspace}
\usepackage{amsthm}
\usepackage{booktabs}
\usepackage{tikz}

\hypersetup{colorlinks,breaklinks,
            linkcolor = purple,
            citecolor = teal,
            urlcolor = purple,
            }

\theoremstyle{definition}
\newtheorem{theorem}{Theorem}[section]

\newtheorem{definition}[theorem]{Definition}
\newtheorem{lemma}[theorem]{Lemma}
\newtheorem{proposition}[theorem]{Proposition}
\newtheorem{def-prop}[theorem]{Definition-Proposition}
\theoremstyle{remark}

\newtheorem{remark}[theorem]{Remark}

\newcommand{\ee}{\mathrm{e}}
\newcommand{\ii}{\mathrm{i}}
\newcommand{\sh}{\operatorname{sh}}

\numberwithin{equation}{section}

\title[Sklyanin--Whittaker integrals]{Determinantal formulas \\ for \\ Sklyanin--Whittaker integrals} 
\author[Taro Kimura]{Taro Kimura}
\address{Universit\'e Bourgogne Europe, CNRS, IMB UMR5584, Dijon, France}
\date{}

\begin{document}

\maketitle

\begin{abstract}
We study multi-variable integrals, that we name Sklyanin--Whittaker integrals, and prove their determinantal formulas. 
We also discuss a $q$-deformation, a determinantal point process, and associated Mellin--Barnes integrals.
\end{abstract}

\setcounter{tocdepth}{1}
\tableofcontents

\section{Introduction and summary}

This paper is devoted to integrals associated with the so-called Sklyanin measure, which was originally introduced by Sklyanin~\cite{Sklyanin:1985} in the context of the quantum integrable system (quantum Toda chain).
Let $G$ be a real reductive group and and let $H_\mathbb{C}$ be its complex Cartan torus.
We denote by $A \subset H_{\mathbb{C}}$ the split real form of the Cartan torus and $\mathfrak{a} = \operatorname{Lie} A \simeq \mathbb{R}^{\operatorname{rk} G}$ its split Cartan subalgebra.
We denote the set of positive (negative, resp.) roots associated with $G$ by $R_G^+$ ($R_G^-$) and set $R_G = R_G^+ \cup R_G^-$.
We denote the Weyl group by $W_G$.
Let $\dd{x}$ be a Lebesgue measure on $\mathfrak{a}$.
Then, the Sklyanin measure of type $G$ is given by
\begin{align}
    \dd{S_G(x)} = 
    \frac{1}{|W_G|} \prod_{\alpha \in R_G^+} \left|\Gamma\left(\frac{\ii \alpha(x)}{2\pi}\right)\right|^{-2} \dd{x} \, ,
\end{align}
where we set $\ii = \sqrt{-1}$ and $\Gamma(\cdot)$ is the gamma function.
Since we only use the root system data for the definition, we apply the root system terminology to classify the Sklyanin measure.
In this paper, we focus on the classical root systems of type $A_{n-1}$, $B_n$, $C_n$, and $D_n$.

Let $\dd{\mu(x)} = w(x) \dd{x}$ be a weighted measure on $\mathbb{R}$, where $w$ decays sufficiently fast at infinity.
Then, we are primarily interested in the following integral, that we call the Sklyanin--Whittaker (SW) integral,
\begin{align}
    Z_G = \frac{1}{|W_G|} \int_{\mathbb{R}^{n}} \prod_{\alpha \in R_G^+} \left| \Gamma \left( \frac{\ii \alpha(x)}{2\pi} \right) \right|^{-2} \prod_{i=1}^{n} \dd{\mu(x_i)} \, .
\end{align}
This type of integrals appears in various contexts of mathematics and physics in addition to the original context in the quantum Toda system, including directed polymers~\cite{OConnell2012}, quantum spin chains~\cite{Kazama:2013rya,Caetano:2020dyp}, and supersymmetric gauge theories~\cite{Honda:2013uca,Hori:2013ika,Fujimori:2015zaa}.

The SW integral is a natural generalization of the well-known Gaussian Unitary Ensemble (GUE) integral (see, e.g.,~\cite{Mehta:2004RMT,Forrester2010LG,Eynard:2015aea}),
\begin{align}
    Z_\text{GUE} = \frac{1}{n!} \int_{\mathbb{R}^n} \prod_{1 \le i < j \le n} |x_j - x_i|^2 \prod_{i=1}^n \ee^{-\frac{1}{2} x_i^2} \dd{x}_i \, ,
\end{align}
which appears particularly in the study of random matrices and related studies. 
Although it is a multi-variable integral, one can explicitly evaluate the GUE integral thanks to the underlying determinantal structure, in particular, the Vandermonde determinant formula, $\prod_{1 \le i < j \le n} (x_j - x_i) = \det_{1 \le i, j \le n} x_i^{j-1}$.
The GUE integral can be straightforwardly generalized to general root systems and also general weight functions instead of the Gaussian weight, which do not spoil the determinantal structure.
The crucial update for the SW integral compared with the GUE integral is the gamma function analogue of the Vandermonde determinant. 
Apparently there is no direct determinantal formula for this case, so that one cannot apply various established techniques of GUE to the SW integral.
However, in this paper, we show a simple trick to convert the SW integral to a determinantal form, which leads to the following result.
\begin{theorem}[Proposition~\ref{prop:SW_det}]
    For a classical root system, the SW integral is given by a determinant.
\end{theorem}
The key idea to obtain a determinantal expression is as follows. 
Using the reflection formula of the gamma function, $\Gamma(z) \Gamma(1-z) = \pi/\sin \pi z$, we have
\begin{align}
    \left|\Gamma\left(\frac{\ii z}{2\pi}\right)\right|^{-2} = \frac{z}{4\pi} \left( \ee^{\tfrac{z}{2}} - \ee^{-\tfrac{z}{2}} \right) \, ,
\end{align}
Hence, the Sklyanin measure is proportional to $\prod_{\alpha \in R_G^+} \alpha(x) \left( \ee^{\tfrac{\alpha(z)}{2}} - \ee^{-\tfrac{\alpha(x)}{2}} \right)$, where we can apply the generalized Vandermonde determinant formulas for classical root systems for both additive and multiplicative variables.
For example, the multiplicative part is given by
\begin{align}
    \prod_{\alpha \in R_G^+} \left( \ee^{\tfrac{\alpha(x)}{2}} - \ee^{-\tfrac{\alpha(x)}{2}} \right) = \ee^{\rho(x)} \prod_{\alpha \in R_G^+} \left( 1 - \ee^{-\alpha(x)} \right) \, ,
\end{align}
where we denote the Weyl vector by $\rho = \frac{1}{2} \sum_{\alpha \in R_G^+} \alpha$.
The Weyl denominator formula gives us the determinantal expression of the multiplicative Vandermonde determinant, which allows us to use Andréief's formula (Lemma~\ref{lem:Andreief_formula}).
This is a sketch of the derivation of the determinantal formula of the SW integral.

\subsection*{Determinantal point process}

An immediate application of the determinantal formula of the SW integral is the determinantal point process.
Based on the Sklyanin measure, we consider the following point process.

Let $\dd{\mu(x)} = w(x) \dd{x}$ be a weighted measure on $\mathbb{R}$ as before, where $w$ decays sufficiently fast at infinity.
For a fixed $n\in\mathbb N$, we define an $n$-point process on $\mathbb R$, that we call the Sklyanin--Whittaker (SW) ensemble of type $G$, given by the following probability measure, 
\begin{align}
 \dd{P_G}(x_1,\ldots,x_n) = \frac{1}{Z_G} \frac{1}{|W_G|} \prod_{\alpha \in R_G^+} \left| \Gamma \left( \frac{\ii \alpha(x)}{2\pi} \right) \right|^{-2} \prod_{i=1}^{n} \dd{\mu(x_i)} \, . 
\end{align}
where $Z_G$ is the corresponding SW integral.
For $G = B_n, C_n, D_n$, we also assume that $w(x) = w(-x)$ so that the measure is invariant under the Weyl group action.

Although it may seem not to be determinantal at first, as discussed above, it turns out that the SW ensemble is categorized into a class of determinantal point processes, called biorthogonal ensemble, introduced by Borodin~\cite{Borodin:1998hxr}, so that it is indeed determinantal.
\begin{proposition}[Proposision~\ref{prop:SW_DDP}]
    The SW ensemble is a determinantal point process.
\end{proposition}
The corresponding correlation kernel is not symmetric for two variables, but it is still self-reproducing.
It would be interesting to study its asymptotic behavior, which we leave for a future study.

\subsection*{$q$-deformation}
In this paper, we also study a $q$-deformed version of the SW integral.
We write the $n$-dimensional torus, $\mathbb{T}^n = \{ z \in (\mathbb{C}^\times)^n\mid |z_i| = 1, i=1,\ldots,n\}$.
Let $q \in \mathbb{C}^\times$ with $|q| < 1$, and $(z;q)_\infty = \prod_{n \ge 0} ( 1 - z q^n)$.
For $z \in \mathbb{T}^{n}$ and a root $\alpha$, we write $z^\alpha = \ee^{\alpha(\log z)}$.
Let $w$ be a function on the unit circle admitting a Fourier expansion, $w(z) = \sum_{n \in \mathbb{Z}} w_n z^n$, with sufficiently fast decay of the Fourier coefficients. 
For $z \in \mathbb{T}$, we define $\dd{\mu(z)} = w(z) \frac{\dd{z}}{2\pi \ii z}$, which is a weighted Haar measure on the unit circle.
Then, we define the following integral, that we call the $q$-Sklyanin--Whittaker ($q$-SW) integral,
\begin{align}
    Z^{(q)}_G = \frac{1}{|W_G|} \int_{\mathbb{T}^{n}} \prod_{\alpha \in R_G} (z^\alpha;q)_\infty \prod_{i=1}^{n} \dd{\mu(z_i)} \, .
\end{align}
This type of integral has a close relation to the so-called $q$-Whittaker function~\cite{Gerasimov2009I,Gerasimov2009II,Gerasimov2011}, which is given by a limit of the Macdonald function~\cite{Macdonald1995}.
\begin{theorem}[Proposition~\ref{prop:q-SW_det}]
    The $q$-SW integral associated with a classical root system is given by a Toeplitz--Hankel type determinant.
\end{theorem}
The idea is the same as before. 
Let $\theta(z;q) = (z;q)_\infty (q/z;q)_\infty$ be the theta function.
Then, we have
\begin{align}
    \prod_{\alpha \in R_G^+} (z^\alpha;q)_\infty (z^{-\alpha};q)_\infty = \prod_{\alpha \in R_G^+} (1 - z^{-\alpha}) \theta(z^{\alpha};q) \, ,
\end{align}
which is a product of the multiplicative and the elliptic Vandermonde determinants.
The elliptic one is indeed the Macdonald denominator associated with the corresponding affine root systems, which has also the determinantal formulas given by Rosengren and Schlosser~\cite{Rosengren2006}.
We remark that there are three more types of reduced affine root systems in addition to the classical $ABCD$ types, i.e., $B_n^\vee$, $C_n^\vee$, and $BC_n$, which are reduced to either $B_n$ or $C_n$ in the limit $q \to 1$.

\subsection*{Mellin--Barnes integral}

Another example that we study in this paper is the Mellin--Barnes integral,
\begin{align}
    \Psi_{G} \left(\begin{matrix} a_1,\ldots,a_r \\ b_1,\ldots,b_s \end{matrix};z\right) = \frac{1}{|W_G|} \int \prod_{\alpha \in R_G} \Gamma(\alpha(x))^{-1} \prod_{\mathsf{w}} \frac{\prod_{\alpha=1}^r \Gamma(\mathsf{w}(x)-a_\alpha)}{\prod_{\alpha=1}^s \Gamma(\mathsf{w}(x)-b_\alpha)} \prod_{i=1}^n w(x_i) \frac{\dd{x}_i}{2\pi\ii} \, ,
\end{align}
where $\{a_\alpha\}_{\alpha=1}^r$ and $\{b_\alpha\}_{\alpha=1}^s$ are parameters and $\mathsf{w}$ is a linear form associated with the irreducible representation of the corresponding Lie group. 
The weight function $w$ is chosen so that the integrand is invariant under the Weyl group action.
This is understood as a multi-variable contour integral, so that the determinantal formula mentioned above cannot be directly applied. 
For the case $n = 1$, it is essentially the ordinary Mellin--Barnes integral, which gives rise to the hypergeometric function. 
\begin{theorem}[Theorems~\ref{thm:MBSW_A}, \ref{thm:MBSW_BCD}]
    The Mellin--Barnes SW integral of a classical root system is given by a Wronskian of the hypergeometric function.
\end{theorem}
We also study a $q$-deformation of the Mellin--Barnes SW integral, for which we prove similar determinantal formulas (Theorems \ref{thm:q-MB_A}, \ref{thm:q-MB_BCD}).

\subsubsection*{Acknowledgments}

We are grateful to O. Khlaif, S. Komatsu, and M. Semenov-Tian-Shansky for useful communications.
This work was supported by EIPHI Graduate School (No.~ANR-17-EURE-0002) and Bourgogne-Franche-Comté region.

\section{Sklyanin--Whittaker integrals}

Let $\dd{\mu(x)} = w(x) \dd{x}$ be a weighted measure on $\mathbb{R}$, where $w$ decays sufficiently fast at infinity.
We define the SW integral of type $G = A_{n-1}, B_n, C_n, D_n$ as follows,
\begin{align}
    Z_G = \frac{1}{|W_G|} \int_{\mathbb{R}^{n}} \prod_{\alpha \in R_G^+} \left| \Gamma \left( \frac{\ii \alpha(x)}{2\pi} \right) \right|^{-2} \prod_{i=1}^{n} \dd{\mu(x_i)} \, .
\end{align}
The Lie algebra data are summarized in \autoref{tab:root_data}.
We remark that $|R_G| = \dim G - \operatorname{rank} G$.
We assume that $w(x) = w(-x)$ for $G = B_n, C_n, D_n$ so that the integrand is invariant under the Weyl group action.
\begin{table}[t]
    \centering
    \begin{tabular}{cccccc} \toprule
        Type $G$ & $\dim G$ &$R_G^+$ & $|R_G|$ & $|W_G|$ & $\rho = \frac{1}{2} \sum_{\alpha \in R_G^+} \alpha$ \\\midrule
        $A_{n-1}$ & $n^2-1$ & $\{ e_i - e_j \}_{1 \le i < j \le n}$ & $n(n-1)$ & $n!$ & $\left(\frac{n-1}{2}, \frac{n-3}{2},\ldots,-\frac{n-1}{2}\right)$ \\ 
        $B_{n}$ & $n(2n+1)$ & $\{ e_i \pm e_j \}_{1 \le i < j \le n} \cup \{e_i\}_{1 \le i \le n}$ & $2n^2$ & $2^n n!$ & $\left(n-\frac{1}{2},n-\frac{3}{2},\ldots,\frac{1}{2}\right)$ \\
        $C_{n}$ & $n(2n+1)$ & $\{ e_i \pm e_j \}_{1 \le i < j \le n} \cup \{2e_i\}_{1 \le i \le n}$ & $2n^2$ & $2^n n!$ & $(n,n-1,\ldots,1)$ \\      
        $D_{n}$ & $n(2n-1)$ & $\{ e_i \pm e_j \}_{1 \le i < j \le n}$ & $2n(n-1)$ & $2^{n-1} n!$ & $(n-1,n-2,\ldots,0)$
        \\ \bottomrule
    \end{tabular}
    \caption{Lie algebra data}
    \label{tab:root_data}
\end{table}
We prove the determinantal formula for this integral.
The following Lemmas are essential.
\begin{lemma}[Andréief's formula]\label{lem:Andreief_formula}
    For families of integrable functions, $\{f_i\}_{i=1,\ldots,n}$ and $\{g_i\}_{i=1,\ldots,n}$, with respect to the measure $\mu$, we have
    \begin{align}
    \frac{1}{n!} \int \det_{1 \le i, j \le n} f_i(x_j) \det_{1 \le i, j \le n} g_i(x_j) \prod_{i=1}^n \dd{\mu(x_i)} = \det_{1 \le i, j \le n} \left( \int f_i(x) g_j(x) \dd{\mu(x)} \right) \, . \label{eq:Andreief_formula}
\end{align}
\end{lemma}
\begin{proof}
    See, e.g.,~\cite{Eynard:2015aea}.
\end{proof}
\begin{lemma}\label{lem:Vandermonde_gamma}
    Let $N_G = \frac{1}{2}|R_G| = |R_G^+|$.
    We write $\sh z = 2 \sinh \left(\dfrac{z}{2}\right) = \ee^{\tfrac{z}{2}} - \ee^{-\tfrac{z}{2}}$.
    Then, we have
    \begin{align}
        \prod_{\alpha \in R_G^+} \left| \Gamma \left( \frac{\ii \alpha(x)}{2\pi} \right) \right|^{-2} = \frac{1}{(4\pi)^{N_G}} \prod_{\alpha \in R_G^+} \alpha(x) \sh \alpha(x) \, .
    \end{align}
\end{lemma}
\begin{proof}
    It follows from the reflection formula of the gamma function, $\Gamma(z) \Gamma(1-z) = \pi / \sin \pi z$.
\end{proof}

We have two types of the product of positive roots for additive and multiplicative variables, which we can write as a determinant.

\noindent
Additive formulas: $\prod_{\alpha \in R_G^+} \alpha(x)$
\begin{subequations}\label{eq:det_additive}
\begin{align}
    A_{n-1} : && \prod_{1 \le i < j \le n} (x_i - x_j) & = \det_{1 \le i, j \le n} x_i^{n-j} \\
    B_n : && \prod_{i=1}^n x_i \prod_{1 \le i < j \le n} (x_i - x_j) (x_i + x_j) & = \det_{1 \le i, j \le n} x_i^{2n - 2j + 1} \\
    C_n : && \prod_{i=1}^n 2x_i \prod_{1 \le i < j \le n} (x_i - x_j) (x_i + x_j) & = 2^n \det_{1 \le i, j \le n} x_i^{2n - 2j + 1} \\
    D_n : && \prod_{1 \le i < j \le n} (x_i - x_j) (x_i + x_j) & = \det_{1 \le i, j \le n} x_i^{2n - 2j} 
\end{align}
\end{subequations}

\noindent
Multiplicative formulas: $\prod_{\alpha \in R_G^+} \sh \alpha(x)$
\begin{subequations}\label{eq:det_multiplicative}
\begin{align}
    A_{n-1} : && \prod_{1 \le i < j \le n} \sh \left( x_i - x_j \right) & = \det_{1 \le i, j \le n} \ee^{\left(\frac{n+1}{2} - j\right)x_i} \\
    B_{n} : && \prod_{i=1}^n \sh x_i \prod_{1 \le i < j \le n} \sh \left( x_i - x_j \right) \sh \left( x_i + x_j \right) & = \det_{1 \le i, j \le n} \left( \ee^{(n + \frac{1}{2} - j) x_i} - \ee^{-(n + \frac{1}{2} - j) x_i} \right) \\ 
    C_{n} : && \prod_{i=1}^n \sh 2x_i \prod_{1 \le i < j \le n} \sh \left( x_i - x_j \right) \sh \left( x_i + x_j \right) & = \det_{1 \le i, j \le n} \left( \ee^{(n + 1 - j) x_i} - \ee^{-(n + 1 - j) x_i} \right) \\ 
    D_{n} : && \prod_{1 \le i < j \le n} \sh \left( x_i - x_j \right) \sh \left( x_i + x_j \right) & = \frac{1}{2} \det_{1 \le i, j \le n} \left( \ee^{(n - j) x_i} + \ee^{-(n - j) x_i} \right) 
\end{align}
\end{subequations}
We remark that the factor appearing in the multiplicative formulas agrees with the corresponding Weyl vector (see \autoref{tab:root_data}).
\begin{proposition}\label{prop:SW_det}
    For $i, j \in \mathbb{Z}_{\ge 0}$, we define 
    \begin{align}
    M_{i,j} = \int_{\mathbb{R}} x^i \ee^{j x} \dd{\mu(x)} \, ,
    \end{align}
    which is finite for any $i$ and $j$ under the assumption that $w$ decays sufficiently fast at infinity.
    Then, the SW integral of classical root systems is given by a determinant as follows,
    \begin{subequations}
    \begin{align}
        Z_{A_{n-1}} & = \frac{1}{(4\pi)^{\frac{n(n-1)}{2}}} \det_{0 \le i, j \le n-1} M_{i,j-\frac{n-1}{2}} \\
        Z_{B_n} & = \frac{1}{(4\pi)^{n^2}} \det_{0 \le i, j \le n-1} \frac{1}{2} \left( M_{2i+1,i+\frac{1}{2}} - M_{2i+1,-i-\frac{1}{2}} \right) \\
        Z_{C_n} & = \frac{2^n}{(4\pi)^{n^2}} \det_{0 \le i, j \le n-1} \frac{1}{2} \left( M_{2i+1,i+1} - M_{2i+1,-i-1} \right) \\
        Z_{D_n} & = \frac{1}{(4\pi)^{n(n-1)}} \det_{0 \le i, j \le n-1} \frac{1}{2} \left( M_{2i,i} + M_{2i,-i} \right) 
    \end{align}
    \end{subequations}
\end{proposition}
\begin{proof}
    By Lemma~\ref{lem:Vandermonde_gamma}, the integrand of the SW integral is written as a product of additive and multiplicative Vandermonde determinants, which are given in \eqref{eq:det_additive} and \eqref{eq:det_multiplicative}.
    Then, applying Andréief's formula (Lemma~\ref{lem:Andreief_formula}), we obtain the result.
\end{proof}

We show an alternative form of the determinantal formula.
Let $p_k$ and $q_k$ be monic polynomials of degree $k$.
Recalling that $\cosh (nx)$ is given by the Chebyshev polynomial of $\cosh(x)$ of degree $n$, $\cosh (nx) = 2^{n-1} \cosh^nx + \cdots$, we have
\begin{subequations}
\begin{align}
    Z_{A_{n-1}} & = \frac{1}{(4\pi)^{\frac{n(n-1)}{2}}} \det_{0 \le i, j \le n-1} \left( \int_{\mathbb{R}} p_i(x) q_j(\ee^x) \dd{\mu_{A_{n-1}}(x)} \right) \, , \\
    Z_{B_n} & = \frac{2^{\frac{(n-1)(n-2)}{2}}}{(4\pi)^{n^2}} \det_{0 \le i, j \le n-1} \left( \int_{\mathbb{R}} p_i(x^2) q_j(\cosh(x)) \dd{\mu_{B_n}(x)} \right) \, , \\
    Z_{C_n} & = \frac{2^{\frac{(n-1)(n-2)}{2}}}{(4\pi)^{n^2}} \det_{0 \le i, j \le n-1} \left( \int_{\mathbb{R}} p_i(x^2) q_j(\cosh(x)) \dd{\mu_{C_n}(x)} \right) \, , \\
    Z_{D_n} & = \frac{2^{\frac{(n-1)(n-2)}{2}}}{(4\pi)^{n(n-1)}} \det_{0 \le i, j \le n-1} \left( \int_{\mathbb{R}} p_i(x^2) q_j(\cosh(x)) \dd{\mu_{D_n}(x)} \right) \, , 
\end{align}
\end{subequations}
where
\begin{subequations}\label{eq:dmu_G}
\begin{gather}
    \dd{\mu_{A_{n-1}}(x)} = \ee^{-\frac{n-1}{2}x} \dd{\mu(x)} \, , \quad 
    \dd{\mu_{B_n}(x)} = x \sinh \left(\frac{x}{2}\right) \dd{\mu(x)} \, , \\ 
    \dd{\mu_{C_n}(x)} = 2 x \sinh \left(x\right) \dd{\mu(x)}
    \, , \quad 
    \dd{\mu_{D_n}(x)} = \dd{\mu(x)} \, .
\end{gather}
\end{subequations}
We remark that, for $G = B_n, C_n, D_n$, the measure $\dd{\mu_G}$ does not depend on the rank, and the integral in the determinant is invariant under the change of variable $x \mapsto -x$ since $w(x) = w(-x)$ in those cases.

\subsection{Gaussian integral}

We consider the simplest example of the SW integral with the normalized Gaussian weight, $w(x) = \ee^{-\tfrac{1}{2}x^2}/\sqrt{2\pi}$, such that $\int_{\mathbb{R}} \dd{\mu(x)} = 1$.
We remark that $w(x) = w(-x)$.
Hence, the integrand is invariant under the Weyl group action for any $G$.

Let $\{H_k\}_{k \ge 0}$ be the monic Hermite polynomials obeying the orthogonality condition,
\begin{align}
    \int_{\mathbb{R}} H_k (x) H_l(x) \ee^{-\tfrac{1}{2} x^2} \frac{\dd{x}}{\sqrt{2\pi}} = k! \delta_{k,l} \, .
\end{align}
We prepare the following Lemmas to evaluate the Gaussian SW integral.
\begin{lemma}\label{lem:Hermite_average}
    For $i,j \in \mathbb{Z}_{\ge 0}$, we have
    \begin{align}
    M^H_{i,j} := \int_{\mathbb{R}} H_i(x) \ee^{j x} \ee^{-\tfrac{1}{2} x^2} \frac{\dd{x}}{\sqrt{2\pi}} = \ee^{\tfrac{j^2}{2}} j^i \, .
    \end{align}
\end{lemma}
\begin{proof}
    It follows from the following identity,
    \begin{align}
        \int_{\mathbb{R}} H_n(x+t) \ee^{-\tfrac{1}{2} x^2} \frac{\dd{x}}{\sqrt{2\pi}} = t^n \, ,
    \end{align}
    which can be obtained by the change of variable $y = x + t$, and then applying the expansion of the generating function, $\ee^{xt - \tfrac{t^2}{2}} = \sum_{k = 0}^\infty H_k(x) \frac{t^k}{k!}$, together with the orthogonality of the Hermite polynomials. 
\end{proof}
\begin{lemma}\label{lem:shifted_Vandermonde}
    Let $\mathsf{G}$ be the Barnes $\mathsf{G}$-function and $a \in \mathbb{C}$. 
    Then, the following identity holds,
    \begin{align}
        \det_{0 \le i, j \le n - 1} (j + a)^{2i} = \frac{2^{(n-1) (n+2a-1)}}{\pi^{\tfrac{n-1}{2}}} \mathsf{G}(n+1) \frac{\mathsf{G}(n+a) \mathsf{G}(n+a+\frac{1}{2}) \mathsf{G}(1+2a)}{\mathsf{G}(1+a) \mathsf{G}(\frac{3}{2}+a) \mathsf{G}(n+2a)} \, .
    \end{align}
\end{lemma}
\begin{proof}
    The determinant is indeed a Vandermonde determinant,
    \begin{align}
        \det_{0 \le i, j \le n - 1} (j + a)^{2i} = \prod_{0 \le i < j \le n-1} ((j+a)^2 - (i+a)^2) \, .
    \end{align}
    Then, applying the functional relation of the $\mathsf{G}$-function, $\mathsf{G}(z+1) = \Gamma(z) \mathsf{G}(z)$, together with the Legendre duplication formula, $\Gamma(2z) = \frac{2^{2z-1}}{\sqrt{\pi}} \Gamma(z) \Gamma\left(z+\frac{1}{2}\right)$, we obtain the result.
\end{proof}
\begin{proposition}
    The SW integral of the Gaussian weight $w(x) = \ee^{-\tfrac{1}{2}x^2}/\sqrt{2\pi}$ is given as follows,
    \begin{subequations}
    \begin{align}
        Z_{A_{n-1}} & = \frac{\ee^{\tfrac{n(n^2-1)}{24}}}{2^{n(n-1)} \pi^{\tfrac{n(n-1)}{2}}} \mathsf{G}(n+1) \\
        Z_{B_n} & = \frac{\ee^{\tfrac{n(4n^2-1)}{24}}}{2^{n(n+1)} \pi^{\tfrac{(n+1)(2n-1)}{2}}} 
        \frac{\mathsf{G}(n+1)\mathsf{G}(n+\frac{3}{2})}{\mathsf{G}(\frac{3}{2})} \\
        Z_{C_n} & = \frac{\ee^{\tfrac{n(n+1)(2n+1)}{12}}}{2^{n(n-1)} \pi^{\tfrac{n(2n+1)}{2}}} 
        \frac{\mathsf{G}(n+1)\mathsf{G}(n+\frac{3}{2})}{\mathsf{G}(\frac{3}{2})} \\
        Z_{D_n} & = \frac{\ee^{\tfrac{n(n-1)(2n-1)}{12}}}{2^{n^2-1} \pi^{\tfrac{(n-1)(2n+1)}{2}}} 
        \frac{\mathsf{G}(n+\frac{1}{2})\mathsf{G}(n+1)}{\mathsf{G}(\frac{3}{2})}
    \end{align}
    \end{subequations}
\end{proposition}
\begin{proof}
    We can use the determinantal formula of the SW integral (Proposition~\ref{prop:SW_det}).
    Since the Hermite polynomial of odd (even, resp.) degree contains only odd (even) degree monomials, the determinant is invariant under the replacement of $M_{i,j}$ by $M_{i,j}^H$.
    Then, we obtain the result by applying Lemma~\ref{lem:Hermite_average} and Lemma~\ref{lem:shifted_Vandermonde}.
\end{proof}
We remark that the exponential factor appearing in the formula is in general written using the Weyl vector.
Let $h^\vee$ be the dual Coxeter number of $G$.
Then, applying the Freudenthal--de Vries strange formula, $12\rho^2 = h^\vee \dim G$, we have $\exp(\frac{1}{2} \rho^2) = \exp(\frac{1}{24} h^\vee \dim G)$.

\subsection{Matrix integral with an external source}

The SW integral has a close relation to the matrix integrals.
We show that it will be identified with the matrix integral with the so-called external source. See, e.g.,~\cite{Brezin:2016eax} for details.

We focus on type $A_{n-1}$.
Let $\mathsf{H}_n = \{ H \in \mathbb{C}^{n \times n} \mid H = H^\dag \}$ be the set of Hermitian matrices of size $n$.
For $a = (a_1,\ldots,a_n) \in \mathbb{R}^n$, we define a diagonal matrix $A = \operatorname{diag}(a_1,\ldots,a_n)$.
We consider a Hermitian matrix integral with an external source $A$,
\begin{align}
    Z_{\mathsf{H}_n}(a) = \int_{\mathsf{H}_n} \ee^{-\tr V(H) + \tr HA} \dd{H} \, ,
\end{align}
where $V$ is the potential function ensuring convergence of the integral.
\begin{lemma}[Harish-Chandra--Itzykson--Zuber (HCIZ) formula of type $A_{n-1}$]
    Let $A = \operatorname{diag}(a_1,\ldots,a_n)$, $B = \operatorname{diag}(b_1,\ldots,b_n)$, and we denote the associated Vandermonde determinant by $\Delta(A) = \Delta(a_1,\ldots,a_n) = \prod_{1 \le i < j \le n}(a_i-a_j)$ and $\Delta(B) = \Delta(b_1,\ldots,b_n) = \prod_{1 \le i < j \le n} (b_i - b_j)$.
    Then, under the normalization of the Haar measure $\int_{\mathrm{U}(n)} \dd{U} = 1$, the following integral formula holds,
    \begin{align}
        \int_{\mathrm{U}(n)} \ee^{\tr U A U^\dag B} \dd{U} = \frac{\mathsf{G}(n+1)}{\Delta(A) \Delta(B)} \det_{1 \le i, j \le n} \ee^{a_i b_j} \, ,
    \end{align}
    where $\mathsf{G}(n+1) = \prod_{k=0}^{n-1} k!$ (Barnes $\mathsf{G}$-function).
\end{lemma}
\begin{proof}
    See, e.g.,~\cite{Eynard:2015aea}.
\end{proof}
Let $x = (x_1,\ldots,x_n) \in \mathbb{R}^n$ be eigenvalues of $H \in \mathsf{H}_n$ and denote the Vandermonde determinant by $\Delta(x) = \prod_{1 \le i < j \le n}(x_i - x_j)$.
Applying the HCIZ formula, we have
\begin{align}
    Z_{\mathsf{H}_n}(a) & = \frac{\mathsf{G}(n+1)}{n! \Delta(A)} \int_{\mathbb{R}^n} \Delta({x}) \det_{1 \le i, j \le n} \ee^{x_i a_j} \prod_{i=1}^n \ee^{-V(x_i)} \dd{x_i} \cr
    & = \frac{\mathsf{G}(n+1)}{\Delta(A)} \det_{1 \le i, j \le n} \left( \int_{\mathbb{R}} x^{i-1} \ee^{-V(x) + x a_j} \dd{x} \right) \, .
\end{align}
\begin{proposition}
    The SW integral of type $A_{n-1}$ agrees with the Hermitian matrix integral with an external source $A = \operatorname{diag}(\rho_i)_{i=1}^n$ with the corresponding Weyl vector $\rho_i=\frac{n+1}{2}-i$ under the identification $\dd{\mu(x)} = \ee^{-V(x)}\dd{x}$.
\end{proposition}
\begin{proof}
Consider the Vandermonde determinant $\Delta(A) = \Delta(a_1,\ldots,a_n)$ for $A = \operatorname{diag}(\rho_i)_{i=1}^n$,
\begin{align}
    \Delta(A) = \prod_{1 \le i<j \le n} (j-i) = \prod_{i=1}^{n-1} i! = \mathsf{G}(n+1) \, .
\end{align}
Then, we have
\begin{align}
    Z_{H_n}(\rho) = \det_{1 \le i, j \le n} \left( \int_{\mathbb{R}} x^{i-1} \ee^{-V(x) + x \rho_j} \dd{x} \right) \, ,
\end{align}
which, up to a constant, agrees with $Z_{A_{n-1}}$ under the identification $\dd{\mu(x)} = \ee^{-V(x)}\dd{x}$.    
\end{proof}
We have a similar structure for $G = B_n, C_n, D_n$.
In these cases, we should consider an integral over the corresponding Lie algebra instead of the Hermitian matrix of the case $G = A_{n-1}$. 
For example, for $G = B_n, D_n$, it is an integral over skew-symmetric matrices, which can be transformed into the normal form via the orthogonal transform.
The corresponding Jacobian will be given by the product over the root system. 
In the presence of an external source, after applying the Harish-Chandra formula, a half of the root product is canceled and the source term is given by the determinant of $\cosh$ or $\sinh$ (see, e.g.,~\cite{Ferrer:2006fd,Forrester2019}), which agrees with Proposition~\ref{prop:SW_det}.

\subsection{Determinantal point process}

Motivated by the discussion above, we study a point process, that we call the Sklyanin--Whittaker ensemble.
\begin{definition}\label{def:SW_ensemble}
Let $\dd{\mu(x)} = w(x) \dd{x}$ be a weighted measure on $\mathbb{R}$, where $w$ decays sufficiently fast at infinity.
For a fixed $n\in\mathbb N$, we define the SW ensemble of type $G$ of $n$ points on $\mathbb{R}$, given by the following probability measure, 
\begin{align}
 \dd{P_G}(x_1,\ldots,x_n) = \frac{1}{Z_G} \frac{1}{|W_G|} \prod_{\alpha \in R_G^+} \left| \Gamma \left( \frac{\ii \alpha(x)}{2\pi} \right) \right|^{-2} \prod_{i=1}^{n} \dd{\mu(x_i)} \, . 
\end{align}
where $Z_G$ is the corresponding SW integral.
For $G = B_n, C_n, D_n$, we also assume that $w(x) = w(-x)$ so that the measure is invariant under the Weyl group action.
\end{definition}
Rewriting the gamma factor in the joint distribution function as $\prod_{\alpha \in R_G^+} \alpha(x) \sinh \left( \tfrac{\alpha(z)}{2} \right)$ together with the determinant formulas \eqref{eq:det_additive} and \eqref{eq:det_multiplicative}, we immediately identify the SW ensemble as a special case of the so-called biorthogonal ensemble introduced by Borodin~\cite{Borodin:1998hxr}, which is a determinantal point process. 

We define a pairing matrix for $i,j \in \mathbb{Z}_{\ge 0}$,
\begin{align}
    M_{i,j}^G = \left< p_i, q_j \right>_G := 
    \begin{cases}
        \displaystyle
        \int_{\mathbb{R}} p_{i}(x) q_{j}(\ee^x) \dd{\mu_{G}(x)} & (G = A_{n-1}) \\[1em]
        \displaystyle
        \int_{\mathbb{R}} p_{i}(x^2) q_{j}(\cosh(x)) \dd{\mu_{G}(x)} & (G = B_n, C_n, D_n)
    \end{cases}
\end{align}
where the measure $\mu_G$ is defined in \eqref{eq:dmu_G}.
By the assumption on the weight $w$, decaying sufficiently fact at infinity, there exists the inverse of this matrix denoted by $\check{M}^G$.
We denote the joint probability density function of the SW ensemble with respect to the measure $\prod_{i=1}^n \dd{\mu_G(x_i)}$ by $p_G(x_1,\ldots,x_n)$, i.e., $\dd{P_G} = p_G(x_1,\dots,x_n) \prod_{i=1}^n \dd{\mu_G(x_i)}$.
Then, we have
\begin{align}
    p_G(x_1,\ldots,x_n) = \frac{1}{n!} \det_{1 \le i, j \le n} K_G(x_i,x_j) \, ,
\end{align}
where the kernel is given by
\begin{align}
    K_G(x,y) = 
    \begin{cases}
        \displaystyle
        \sum_{0 \le i, j \le n-1} p_i(x) \check{M}^G_{i,j} q_j(\ee^y) & (G = A_{n-1}) \\[.5em]
        \displaystyle 
        \sum_{0 \le i, j \le n - 1} p_i(x^2) \check{M}^G_{i,j} q_j(\cosh(y)) & (G = B_n, C_n, D_n)
    \end{cases}
\end{align}
This kernel is not symmetric, but self-reproducing,
\begin{align}
    \int_{\mathbb{R}} K_G(x,x) \dd{\mu_G(x)} = n \, , \qquad
    \int_{\mathbb{R}} K_G(x,y) K_G(y,z) \dd{\mu_G(y)} = K_G(x,z) \, .
\end{align}
We define the $k$-point correlation function of the SW ensemble,
\begin{align}
    \rho_{G,k}(x_1,\ldots,x_k) := \frac{n!}{(n-k)!} \int_{\mathbb{R}^{n-k}} p_G(x_1,\ldots,x_n) \prod_{i=k+1}^n \dd{\mu_G(x_i)} \, .
\end{align}
Then, we have the following formula.
\begin{proposition}\label{prop:SW_DDP}
    The correlation function of the SW ensemble is given by a determinant of the corresponding kernel, 
    \begin{align}
        \rho_{G,k}(x_1,\ldots,x_k) = \det_{1 \le i, j \le k}  K_G(x_i,x_j) \, ,
    \end{align}
    i.e., it is a determinantal point process,
\end{proposition}
\begin{proof}
    It directly follows from the self-reproducing property of the kernel.
\end{proof}

\subsection{$q$-Sklyanin--Whittaker integral}

Let $w$ be a function on the unit circle $\mathbb{T}$ admitting a Fourier expansion, $w(z) = \sum_{n \in \mathbb{Z}} w_n z^n$, with sufficiently fast decay of the Fourier coefficients.
We assume $w_n = w_{-n}$ for $G = B_n, C_n, D_n$.
For $z \in \mathbb{T}=\{z \in \mathbb{C}^\times \mid |z| = 1 \}$, we define $\dd{\mu(z)} = w(z) \frac{\dd{z}}{2\pi \ii z}$, which is a weighted Haar measure on the unit circle.
We define the $q$-SW integral as follows,
\begin{align}
    Z^{(q)}_G = \frac{1}{|W_G|} \int_{\mathbb{T}^{n}} \prod_{\alpha \in R_G} (z^\alpha;q)_\infty \prod_{i=1}^{n} \dd{\mu(z_i)} \, .
\end{align}
By the definition of the theta function \eqref{eq:theta_def}, we may write
\begin{align}
    Z^{(q)}_G = \frac{1}{|W_G|} \int_{\mathbb{T}^{n}} \prod_{\alpha \in R_G^+} (1-z^{-\alpha}) \theta(z^{\alpha};q) \prod_{i=1}^{n} \dd{\mu(z_i)} \, .
\end{align}
As in the previous case, we have two types of Vandermonde determinants in the integrand.
We write $\theta(x) = \theta(x;q)$ and define the following products of theta functions, that we call the elliptic Vandermonde determinants,
\begin{subequations}
\begin{align}
    W_{A_{n-1}}(x) & = \prod_{1 \le i < j \le n} x_j \theta(x_i/x_j) \\
    W_{B_n}(x) & = \prod_{i=1}^n \theta(x_i) \prod_{1 \le i < j \le n} x_i^{-1} \theta(x_i/x_j) \theta(x_i x_j) \\
    W_{C_n}(x) & = \prod_{i=1}^n x_i^{-1} \theta(x_i^2) \prod_{1 \le i < j \le n} x_i^{-1} \theta(x_i/x_j) \theta(x_i x_j) \\
    W_{D_n}(x) & = \prod_{1 \le i < j \le n} x_i^{-1} \theta(x_i/x_j) \theta(x_i x_j) 
\end{align}
\end{subequations}
There is an elliptic generalization of the Vandermonde determinant formula shown by Rosengren and Schlosser.
\begin{proposition}[Rosengren--Schlosser~\cite{Rosengren2006}]\label{prop:det_elliptic}
The following determinant formulas hold,
\begin{subequations}
\begin{align}
    \det_{1 \le i, j \le n} \left( x_i^{j-1} \theta( (-1)^{n-1} q^{j-1} t x_i^n ; q^n) \right) & = \frac{(q;q)_\infty^n}{(q^n;q^n)_\infty^n} \theta(t x_1 \cdots x_n) W_{A_{n-1}}(x) \\
    \det_{1 \le i, j \le n} \left( x_i^{j-n} \theta(q^{j-1} x_i^{2n-1};q^{2n-1}) - x_i^{n+1-j} \theta(q^{j-1}x_i^{1-2n};q^{2n-1}) \right)
    & = \frac{2 (q;q)_\infty^n}{(q^{2n-1};q^{2n-1})_\infty^n} W_{B_n}(x) \\
    \det_{1 \le i, j \le n} \left( x_i^{j-n-1} \theta(-q^{j} x_i^{2n+2};q^{2n+2}) - x_i^{n+1-j} \theta(-q^{j}x_i^{-2n-2};q^{2n+2}) \right)
    & = \frac{(q;q)_\infty^n}{(q^{2n+2};q^{2n+2})_\infty^n} W_{C_n}(x) \\
    \det_{1 \le i, j \le n} \left( x_i^{j-n} \theta(-q^{j-1} x_i^{2n-2};q^{2n-2}) + x_i^{n-j} \theta(-q^{j-1}x_i^{2-2n};q^{2n-2}) \right)
    & = \frac{4(q;q)_\infty^n}{(q^{2n-2};q^{2n-2})_\infty^n} W_{D_n}(x)
\end{align}
\end{subequations}
\end{proposition}
For $G = A_{n-1}$, we have an auxiliary parameter $t$, which is called the norm in \cite{Rosengren2006}.
We assume $|t| < 1$ to apply the following Lemma concerning the theta function inverse.
\begin{lemma}\label{lem:theta_expansion}
    For $|q| < |z| < 1$, we have
    \begin{align}
        \frac{1}{\theta(z;q)} = \frac{1}{(q;q)_\infty^2} \left( \sum_{n \ge 0} \frac{(-1)^n q^{n+1 \choose 2}}{1 - z q^r} - z^{-1} \sum_{n \ge 0} \frac{(-1)^n q^{n+1 \choose 2}}{1 - z^{-1} q^r} \right) = \sum_{m \in \mathbb{Z}} c_m z^m \, ,
    \end{align}
    where
    \begin{align}
        c_{m_{\ge 0}} = \frac{1}{(q;q)_\infty^2} \sum_{n \ge 0} (-1)^n q^{n+1 \choose 2} q^{mn} \, , \qquad 
        c_{m_{< 0}} = \frac{1}{(q;q)_\infty^2} \sum_{n \ge 0} (-1)^n q^{n \choose 2} q^{(|m|-1)(n+1)} \, .
    \end{align}
\end{lemma}
\begin{proof}
    See \cite[Lemma 1]{Andrews1984}.
\end{proof}
\begin{proposition}\label{prop:q-SW_det}
\begin{subequations}
The $q$-SW integral of type $A_{n-1}$, $B_n$, $C_n$, and $D_n$ is given by a Toeplitz--Hankel determinant,
\begin{align}
    Z_{A_{n-1}}^{(q)} & = \frac{1}{(q;q)_\infty^n} \det_{1 \le i, j \le n} \sum_{k, m \in \mathbb{Z}} (-1)^{nm} q^{(j-1) m + n {m \choose 2}} t^{k+m} c_k w_{i-j-nm-k} \\
    Z_{B_n}^{(q)} & = \frac{1}{2(q;q)_\infty^n} \det_{1 \le i, j \le n} \sum_{m \in \mathbb{Z}} (-1)^m q^{(j-1) m+(2n+1) {m \choose 2}} \left( w_{i-j-m(2n+1)} - w_{2n+1-i-j-m(2n+1)} \right) \\
    Z_{C_n}^{(q)} & = \frac{1}{(q;q)_\infty^n} \det_{1 \le i, j \le n} \sum_{m \in \mathbb{Z}} q^{j m+(2n+2) {m \choose 2}} \left( w_{i-j-m(2n+2)} - w_{2n+2-i-j-m(2n+2)} \right) \\
    Z_{D_n}^{(q)} & = \frac{1}{4 (q;q)_\infty^n} \det_{1 \le i, j \le n} \sum_{m \in \mathbb{Z}} q^{(j-1)m+(2n-2) {m \choose 2}} \left( w_{i-j-m(2n-2)} + w_{2n-i-j-m(2n-2)} \right) \, .
\end{align}
\end{subequations}    
\end{proposition}
\begin{proof}
Together with the multiplicative determinantal formula \eqref{eq:det_multiplicative}, the Vandermonde factors can be computed as follows,
\begin{subequations}\label{eq:q-SW_det}
\begin{align}
    A_{n-1} : \qquad & \prod_{1 \le i < j \le n} \left( 1 - \frac{z_j}{z_i} \right) \theta \left( \frac{z_i}{z_j} \right) 
    = W_{A_{n-1}}(z) \det_{1 \le i, j \le n} z_i^{-j+1} 
    \, , \\
    B_{n} : \qquad & \prod_{i=1}^n (1-z_i^{-1}) \theta(z_i) \prod_{1 \le i < j \le n} \left( 1 - \frac{z_j}{z_i} \right) \left( 1 - \frac{1}{z_i z_j} \right) \theta \left( \frac{z_i}{z_j} \right) \theta (z_i z_j) \nonumber \\
    & = \prod_{i=1}^n z_i^{-\tfrac{1}{2}} W_{B_n}(z) \det_{1 \le i, j \le n} \left( z_i^{n + \frac{1}{2} - j} - z_i^{- n - \frac{1}{2} + j} \right) \, , \\
    C_{n} : \qquad & \prod_{i=1}^n (1-z_i^{-2}) \theta(z_i^2) \prod_{1 \le i < j \le n} \left( 1 - \frac{z_j}{z_i} \right) \left( 1 - \frac{1}{z_i z_j} \right) \theta \left( \frac{z_i}{z_j} \right) \theta (z_i z_j) \nonumber \\
    & = W_{C_n}(z) \det_{1 \le i, j \le n} \left( z_i^{n + 1 - j} - z_i^{- n - 1 + j} \right) \, , \\
    D_{n} : \qquad & \prod_{1 \le i < j \le n} \left( 1 - \frac{z_j}{z_i} \right) \left( 1 - \frac{1}{z_i z_j} \right) \theta \left( \frac{z_i}{z_j} \right) \theta (z_i z_j) \nonumber \\
    & = \frac{1}{2} W_{D_n}(z) \det_{1 \le i, j \le n} \left( z_i^{n - j} + z_i^{- n + j} \right) \, .
\end{align}
\end{subequations}
Then, by Proposition~\ref{prop:det_elliptic}, we can apply Andréief's formula to rewrite the $n$-variable integral as a determinant.
Each element of the corresponding matrix is a single-variable contour integral on $\mathbb{T}$, yielding the constant term of the integrand.
By the series expansion of the theta function \eqref{eq:theta_def} (together with Lemma~\ref{lem:theta_expansion} for type $A_{n-1}$), we obtain the result.
We also remark that $w_n = w_{-n}$ for $G = B_n, C_n, D_n$.
The infinite series appearing in the determinant converge since the Fourier coefficients of $w$ decay sufficiently fast.
\end{proof}

In the limit $q \to 0$, the integral is reduced to the ordinary Cartan torus integral associated with the classical groups,
\begin{align}
    Z^{(q=0)}_G = \frac{1}{|W_G|} \int_{\mathbb{T}^{n}} \prod_{\alpha \in R_G} (1-z^\alpha) \prod_{i=1}^{n} \dd{\mu(z_i)} \, ,
\end{align}
for which the Toeplitz--Hankel type determinantal formulas are known~\cite{Johansson1997,Baik2001,Garcia-Garcia:2019uve}.
The formula shown in Proposition~\ref{prop:q-SW_det} is a natural generalization of these formulas.

\section{Mellin--Barnes integral}

Set $r,s,n \in \mathbb{Z}_{\ge 0}$ with $0 \le s \le r$ and $0 \le n \le r$.
Let $I$ be an injective map, $I : \{1,\ldots,n\} \hookrightarrow \{1,\ldots,r \}$.
For $G = A_{n-1}, B_n, C_n$, and $D_n$, we denote by $\operatorname{wt}(\mathcal{R})$ the set of weights of the irreducible representation $\mathcal{R}$ of the corresponding Lie group.
For the defining/vector representation, we have
\begin{align}
    \operatorname{wt}(\text{def}) = 
    \begin{cases}
        \{e_i\}_{i=1,\ldots,n} & (A_{n-1}) \\
        \{\pm e_i\}_{i=1,\ldots,n} \cup \{0\} & (B_n) \\
        \{\pm e_i\}_{i=1,\ldots,n} & (C_n) \\
        \{\pm e_i\}_{i=1,\ldots,n} & (D_n) \\
    \end{cases}
\end{align}
We define the Mellin--Barnes SW integral of type $G$ as follows,
\begin{align}
    \Psi_{G,I} \left(\begin{matrix} a_1,\ldots,a_r \\ b_1,\ldots,b_s \end{matrix};z\right) = \frac{1}{|W_G|} \int \prod_{\alpha \in R_G} \Gamma(\alpha(x))^{-1} \prod_{\mathsf{w} \in \operatorname{wt}(\text{def})} \frac{\prod_{\alpha=1}^r \Gamma(\mathsf{w}(x)-a_\alpha)}{\prod_{\alpha=1}^s \Gamma(\mathsf{w}(x)-b_\alpha)} \prod_{i=1}^n w(x_i) \frac{\dd{x}_i}{2\pi\ii} \, ,
\end{align}
where we define with a parameter $z$ as
\begin{align}
    w(x) = 
    \begin{cases}
        z^{-x} & (A_{n-1}) \\
        z^{-x} + z^{x} & (B_n, C_n, D_n)
    \end{cases} 
\end{align}
The integrand is invariant under the Weyl group action.
We assume that the parameters $\{a_\alpha\}_{\alpha=1,\ldots,r}$ and $\{b_\alpha\}_{\alpha=1,\ldots,s}$ are generic, such that none of the poles of $\Gamma(\pm x-a_\alpha)$ are canceled by those of $\Gamma(\pm x-b_\alpha)$.
The integration contour is given as follows:
Let $\mathcal{C}_\alpha^{\pm}$ be a contour enclosing the poles of $\Gamma(\pm x - a_\alpha)$ located at $\pm (a_\alpha - m)$, $m \in \mathbb{Z}_{\ge 0}$.
We define the $n$-variable integral contour $\mathcal{C}_I = \mathcal{C}_{I(1)}^+ \times \cdots \times \mathcal{C}_{I(n)}^+$.
Since the integrand is invariant under the symmetric group action, the contours associated with $I$ and $I \circ \sigma$ yield the same integral for $\sigma \in \mathfrak{S}_n$.
Then, the integral is understood as follows:
\begin{subequations}
\begin{itemize}
    \item $G = A_{n-1}$ :
    \begin{align}
        \Psi_{G,I} & = \sum_{\sigma \in \mathfrak{S}_n} \frac{1}{|W_G|} \int_{\mathcal{C}_{I \circ \sigma}} \prod_{\alpha \in R_G} \Gamma(\alpha(x))^{-1} \prod_{\mathsf{w} \in \operatorname{wt}(\text{def})} \frac{\prod_{\alpha=1}^r \Gamma(\mathsf{w}(x)-a_\alpha)}{\prod_{\alpha=1}^s \Gamma(\mathsf{w}(x)-b_\alpha)} \prod_{i=1}^n z^{-x_i} \frac{\dd{x}_i}{2\pi\ii} \nonumber \\
        & = \int_{\mathcal{C}_{I}} \prod_{\alpha \in R_G} \Gamma(\alpha(x))^{-1} \prod_{\mathsf{w} \in \operatorname{wt}(\text{def})} \frac{\prod_{\alpha=1}^r \Gamma(\mathsf{w}(x)-a_\alpha)}{\prod_{\alpha=1}^s \Gamma(\mathsf{w}(x)-b_\alpha)} \prod_{i=1}^n z^{-x_i} \frac{\dd{x}_i}{2\pi\ii}
    \end{align}
    \item $G = B_n, C_n, D_n$ :
    \begin{align}
    & \Psi_{G,I} \nonumber \\ 
    & = \sum_{\substack{s_1,\ldots,s_n = \pm 1\\ \sigma \in \mathfrak{S}_n}} \frac{1}{|W_G|} \int_{\mathcal{C}_{I\circ\sigma(1)}^{s_1} \times \cdots \times \mathcal{C}_{I\circ\sigma(n)}^{s_n}} \prod_{\alpha \in R_G} \Gamma(\alpha(x))^{-1} \prod_{\mathsf{w} \in \operatorname{wt}(\text{def})} \frac{\prod_{\alpha=1}^r \Gamma(\mathsf{w}(x)-a_\alpha)}{\prod_{\alpha=1}^s \Gamma(\mathsf{w}(x)-b_\alpha)} \prod_{i=1}^n z^{- s_i x_i} \frac{\dd{x}_i}{2\pi\ii} \nonumber \\
    & = \frac{2^n n!}{|W_G|} \int_{\mathcal{C}_I} \prod_{\alpha \in R_G} \Gamma(\alpha(x))^{-1} \prod_{\mathsf{w} \in \operatorname{wt}(\text{def})} \frac{\prod_{\alpha=1}^r \Gamma(\mathsf{w}(x)-a_\alpha)}{\prod_{\alpha=1}^s \Gamma(\mathsf{w}(x)-b_\alpha)} \prod_{i=1}^n z^{- x_i} \frac{\dd{x}_i}{2\pi\ii} \, .
\end{align}
\end{itemize}
\end{subequations}
Hence, we use the same contour and $w(x) = z^{-x}$ for any $G = A_{n-1}, B_n, C_n$, and $D_n$.
By construction, we have $\Psi_I = \Psi_{I \circ \sigma}$ for any $\sigma \in \mathfrak{S}_n$, i.e., it depends only on the image of $I$.
In other words, we may interpret $I : \{1,\ldots,n\}/\mathfrak{S}_n \to \{1,\ldots,r\}$.
Since $I$ is injective, each contour is taken by no more than single variable. 
If the same contour is applied for different variables, the residue will be zero due to the product over the roots, $\prod_{\alpha \in R_G} \Gamma(\alpha(x))^{-1}$.
The integral converges when $|z| < 1$.

\subsection{Type $A_{n-1}$}

We study the following Mellin--Barnes SW integral of type $A_{n-1}$,
\begin{align}
    \Psi_I(z) = \Psi_{A_{n-1},I} \left(\begin{matrix} a_1,\ldots,a_r \\ b_1,\ldots,b_s \end{matrix};z\right) 
    = \int_{\mathcal{C}_I} \prod_{1 \le i \neq j \le n} \Gamma(x_i-x_j)^{-1} \prod_{i=1}^n \frac{\prod_{\alpha=1}^r \Gamma(x_i-a_\alpha)}{\prod_{\alpha=1}^s \Gamma(x_i-b_\alpha)} z^{-x_i} \frac{\dd{x_i}}{2\pi\ii} \, .
\end{align}
As before, we assume that $0 \le s \le r$, $0 \le n \le r$, and parameters $\{a_\alpha\}_{\alpha=1,\ldots,r}$ and $\{b_\alpha\}_{\alpha=1,\ldots,s}$ are generic, and also $|z| < 1$ so that the integral converges.

We define a single-variable integral, which will be a building block of the $n$-variable integral, 
\begin{align}
    \psi_\alpha(z) = \psi_\alpha\left(\begin{matrix} a_1,\ldots,a_r \\ b_1,\ldots,b_s \end{matrix};z\right) = \int_{\mathcal{C}^+_\alpha} \frac{\prod_{\alpha=1}^r \Gamma(x-a_\alpha)}{\prod_{\alpha=1}^s \Gamma(x-b_\alpha)} z^{-x} \frac{\dd{x}}{2\pi\ii} \, , \quad \alpha = 1, \ldots, r \, .
\end{align}
We write a differential operator, $d_z = z \pdv{}{z}$.
Then, $\{\psi_\alpha\}_{\alpha=1,\ldots,r}$ forms a basis of the solution space of the following differential equation,
\begin{align}
    \left[ \prod_{\beta=1}^r \left( d_z + a_\beta \right) - (-1)^{p+q} z \prod_{\beta=1}^s \left( d_z + b_\beta + 1 \right) \right] y(z) = 0 \, .
\end{align}
\begin{proposition}
    Under the assumption above, we have
    \begin{align}
        \psi_\alpha(z) = \frac{\prod_{\beta (\neq \alpha)}^r \Gamma(a_\alpha - a_\beta)}{\prod_{\beta=1}^s \Gamma(a_\alpha-b_\beta)} z^{-a_\alpha} {}_s F_{r-1} \left( \begin{matrix} \{1+b_\beta-a_\alpha\}_{\beta=1}^s \\ \{1+a_\beta-a_\alpha\}_{\beta(\neq \alpha)}^r \end{matrix} ; (-1)^{r+s} z \right) \, .
    \end{align}    
\end{proposition}
\begin{proof}
    Recalling that the residue of the gamma function is given by
    \begin{align}
        \operatorname*{Res}_{x = m} \Gamma(-x) = \frac{(-1)^m}{m!} \, , \quad m \in \mathbb{Z}_{\ge 0} \, ,
    \end{align}
    we evaluate the integral by summing up all the residue contributions,
    \begin{align}
        \psi_\alpha(z) & = \sum_{m = 0}^\infty z^{-a_\alpha + m} \frac{(-1)^m}{m!} \frac{\prod_{\beta(\neq \alpha)}^r \Gamma(a_\alpha - a_\beta - m)}{\prod_{\beta=1}^s \Gamma(a_\alpha - b_\beta - m)} \nonumber \\
        & = \frac{\prod_{\beta (\neq \alpha)}^r \Gamma(a_\alpha - a_\beta)}{\prod_{\beta=1}^s \Gamma(a_\alpha-b_\beta)} z^{-a_\alpha} \sum_{m=0}^\infty \frac{((-1)^{r+s} z)^m}{m!} \frac{\prod_{\beta=1}^s(1+b_\beta-a_\alpha)_m}{\prod_{\beta(\neq\alpha)}^r(1+a_\beta-a_\alpha)_m} \, ,
    \end{align}
    where we used $\Gamma(x)/\Gamma(x - m) = (-1)^m (1-x)_m$ for $m \in \mathbb{Z}_{\ge 0}$.
\end{proof}
We define
\begin{align}
\psi_{A_{n-1},\alpha}(z) & = \ee^{(n-1) \pi \ii a_\alpha} \psi_\alpha((-1)^{n-1} z) \, .
\end{align}
Then, we have the following formula.
\begin{theorem}\label{thm:MBSW_A}
    The Mellin-Barnes SW integral of type $A_{n-1}$ is given by a Wronskian,
    \begin{align}
        \Psi_I(z) = \prod_{1 \le i < j \le n} \frac{\sin \pi (a_{I(i)} - a_{I(j)})}{\pi} \det_{1 \le i, j \le n} d_z^{i-1} \psi_{A_{n-1},I(j)}(z) \, .
    \end{align}
\end{theorem}
\begin{proof}
We consider a generalized version of the integral,
\begin{align}
    \tilde{\Psi}_I \left(\begin{matrix} a_1,\ldots,a_r \\ b_1,\ldots,b_s \end{matrix};z_1,\ldots,z_n\right) & = \int_{\mathcal{C}_I} \prod_{1 \le i \neq j \le n} \Gamma(x_i-x_j)^{-1} \prod_{i=1}^n \frac{\prod_{\alpha=1}^r \Gamma(x_i-a_\alpha)}{\prod_{\alpha=1}^s \Gamma(x_i-b_\alpha)} z_i^{-x_i} \frac{\dd{x_i}}{2\pi\ii} \nonumber \\
    & = \int_{\mathcal{C}_I} \prod_{1 \le i < j \le n} (x_i - x_j) \frac{\sin \pi (x_i - x_j)}{\pi} \prod_{i=1}^n \frac{\prod_{\alpha=1}^r \Gamma(x_i-a_\alpha)}{\prod_{\alpha=1}^s \Gamma(x_i-b_\alpha)} z_i^{-x_i} \frac{\dd{x_i}}{2\pi\ii} \, .
\end{align}
We assume that $|z_i| < 1$ for any $i \in \{1,\ldots,n\}$ so that the integral converges.
We remark that, in contrast to the original integral $\Phi_I$, the generalized version $\tilde{\Psi}_I$ does not necessarily agrees with $\tilde{\Psi}_{I\circ\sigma}$ for $\sigma \in \mathfrak{S}_n$.
Then, we have the following,
\begin{align}
    \tilde{\Psi}_I 
    & = \sum_{0 \le m_1,\ldots,m_n \le \infty} \prod_{1 \le i < j \le n} (a_{I(i)} - a_{I(j)} - m_i + m_j) \frac{\sin \pi (a_{I(j)} - a_{I(i)} - m_j + m_i)}{\pi} \nonumber \\
    & \qquad \times \prod_{i=1}^n z_i^{-a_{I(i)} + m_i} \frac{(-1)^{m_i}}{m_i!} \frac{\prod_{\beta(\neq I(i))}^r \Gamma(a_{I(i)} - a_\beta - m_i)}{\prod_{\beta=1}^s \Gamma(a_{I(i)} - b_\beta - m_i)} \nonumber \\
    & = \prod_{1 \le i < j \le n} \frac{\sin \pi (a_{I(j)} - a_{I(i)})}{\pi} \prod_{i=1}^n \ee^{(n-1)\pi\ii a_{I(i)}} \sum_{0 \le m_1,\ldots,m_n \le \infty} \prod_{1 \le i < j \le n} (a_{I(i)} - a_{I(j)} - m_i + m_j) \nonumber \\
    & \qquad \times \prod_{i=1}^n ((-1)^{n-1} z_i)^{-a_{I(i)} + m_i} \frac{(-1)^{m_i}}{m_i!} \frac{\prod_{\beta(\neq I(i))}^r \Gamma(a_{I(i)} - a_\beta - m_i)}{\prod_{\beta=1}^s \Gamma(a_{I(i)} - b_\beta - m_i)} \nonumber \\
    & = \prod_{1 \le i < j \le n} \frac{\sin \pi (a_{I(j)} - a_{I(i)})}{\pi} \left( d_{z_j} - d_{z_i} \right) \prod_{i=1}^n \psi_{A_{n-1},I(i)}(z_i) \nonumber \\
    & = \prod_{1 \le i < j \le n} \frac{\sin \pi (a_{I(j)} - a_{I(i)})}{\pi} \det_{1 \le i, j \le n} \left( d_{z_i}^{j-1} \psi_{A_{n-1},I(i)}(z_i) \right)  \, .
\end{align}
    Taking the limit $z_i \to z$ for all $i \in \{1,\ldots,n\}$, we obtain the result.
\end{proof}

\subsection{Type $B_n$, $C_n$, $D_n$}

For $G = B_n, C_n, D_n$, we consider the following integral,
\begin{align}
    \Psi_I(z) = \frac{2^n n!}{|W_G|} \int_{\mathcal{C}_I} \prod_{\alpha \in R_G} \Gamma(\alpha(x))^{-1} \prod_{\mathsf{w} \in \operatorname{wt}(\text{def})} \frac{\prod_{\alpha=1}^r \Gamma(\mathsf{w}(x)-a_\alpha)}{\prod_{\alpha=1}^s \Gamma(\mathsf{w}(x)-b_\alpha)} \prod_{i=1}^n z^{- x_i} \frac{\dd{x}_i}{2\pi\ii} \, .
\end{align}
We write $\Gamma(\pm x-a) = \Gamma(+ x-a_\alpha) \Gamma(- x-a_\alpha)$.

We define a single-variable integral,
\begin{align}
    {\psi}^{\pm}_\alpha(z) = \int_{\mathcal{C}^+_\alpha} \frac{\prod_{\alpha=1}^r \Gamma(\pm x-a_\alpha)}{\prod_{\alpha=1}^s \Gamma(\pm x-b_\alpha)} z^{- x} \frac{\dd{x}}{2\pi\ii} \, , \qquad \alpha = 1, \ldots, r \, .
\end{align}
They yield $r$ independent solutions of the differential equation,
\begin{align}
    \left[ \prod_{\beta=1}^r \left( d_z + a_\beta \right) \prod_{\beta = 1}^s (d_z - b_\beta - 1) - (-1)^{r+s} z \prod_{\beta = 1}^r( d_z - a_\alpha ) \prod_{\beta=1}^s \left( d_z + b_\beta + 1 \right) \right] y = 0 \, .
\end{align}
In order to discuss all the independent solutions, we instead consider the following,
\begin{align}
    \tilde{\psi}^{\pm}_\alpha(z) = \int_{\mathcal{C}^+_\alpha} \frac{\prod_{\alpha=1}^r \Gamma(x-a_\alpha) \prod_{\alpha=1}^s \Gamma(x + b_\alpha + 1)}{\prod_{\alpha=1}^r \Gamma(x + a_\alpha + 1) \prod_{\alpha=1}^s \Gamma(x-b_\alpha)} z^{- x} \frac{\dd{x}}{2\pi\ii} \, , \qquad \alpha = 1, \ldots, r + s \, ,    
\end{align}
where we define additional contours, $\mathcal{C}_{r+\alpha}^+$, $\alpha = 1,\ldots,s$, enclosing the poles of $\Gamma(x + b_\alpha + 1)$ at $-b_\alpha + 1 - m$, $m \in \mathbb{Z}_{\ge 0}$.
The integrand of $\psi^\pm$ and $\tilde{\psi}^\pm$ is converted to each other using the reflection formula of the gamma function with additional sine factors.
\begin{proposition}
    Under the notation and the assumption above, we have
    \begin{align}
        \psi^\pm_\alpha(z) & = \frac{\prod_{\beta (\neq \alpha)}^r \Gamma(a_\alpha - a_\beta) \prod_{\beta = 1}^r \Gamma(-a_\alpha - a_\beta)}{\prod_{\beta = 1}^s \Gamma(a_\alpha - b_\beta) \Gamma(-a_\alpha - b_\beta)} \nonumber \\
        & \times z^{-a_\alpha} {}_{r+s} F_{r+s-1} \left( \begin{matrix}
            \{-a_\alpha-a_\beta\}_{\beta=1}^r , \{1+b_\beta-a_\alpha\}_{\beta=1}^s \\ \{1+a_\beta-a_\alpha\}_{\beta(\neq\alpha)}^r , \{-a_\alpha-b_\beta\}_{\beta=1}^s
        \end{matrix} ; (-1)^{r+s} z \right) \, .
    \end{align}
\end{proposition}
\begin{proof}
    Evaluating the residues of the poles, we have
    \begin{align}
        \psi^\pm_\alpha(z) 
        & = \sum_{m = 0}^\infty z^{-a_\alpha + m} \frac{(-1)^m}{m!} \frac{\prod_{\beta (\neq \alpha)}^r \Gamma(a_\alpha - a_\beta - m) \prod_{\beta = 1}^r \Gamma(-a_\alpha - a_\beta + m)}{\prod_{\beta = 1}^s \Gamma(a_\alpha - b_\beta - m) \Gamma(-a_\alpha - b_\beta + m)} \nonumber \\
        & = \frac{\prod_{\beta (\neq \alpha)}^r \Gamma(a_\alpha - a_\beta) \prod_{\beta = 1}^r \Gamma(-a_\alpha - a_\beta)}{\prod_{\beta = 1}^s \Gamma(a_\alpha - b_\beta) \Gamma(-a_\alpha - b_\beta)} \nonumber \\
        & \qquad \times z^{-a_\alpha} \sum_{m = 0}^\infty \frac{((-1)^{p+q} z)^m}{m!} \frac{\prod_{\alpha = 1}^r (-a_\alpha-a_\beta)_m \prod_{\beta=1}^s (1+b_\beta-a_\alpha)_m}{\prod_{\beta(\neq\alpha)}^r (1+a_\beta-a_\alpha)_m \prod_{\beta=1}^s (-a_\alpha-b_\beta)_m} \, ,
    \end{align}
    from which we conclude the proof.
\end{proof}
We define
\begin{align}
    \psi_{G,\alpha}(z) & = 
    \begin{cases}
        \ee^{\pi \ii a_\alpha} \frac{\prod_{\beta=1}^r \Gamma(-a_\beta)}{\prod_{\beta=1}^s \Gamma(-b_\beta)} \psi_\alpha^\pm (-z) & (G = B_n) \\
        \psi_\alpha^\pm (z) & (G = C_n, D_n)
    \end{cases}
\end{align}
Then, we have the following result.
\begin{theorem}\label{thm:MBSW_BCD}
    Let $a_I = \{a_{I(i)}\}_{i=1,\ldots,n}$.
    The Mellin--Barnes SW integral of type $B_n$, $C_n$, $D_n$ is given by a Wronskian,
    \begin{align}
        \Psi_{G,I}(z) = \prod_{\alpha \in R_G^+} \frac{\sin \pi \alpha(a_I)}{\pi} \times
        \begin{cases}
            \displaystyle \det_{1 \le i, j \le n} d_z^{2j-1} \psi_{G,I(i)}(z) & (G = B_n) \\[.8em]
            \displaystyle 2^n \det_{1 \le i, j \le n} d_z^{2j-1} \psi_{G,I(i)}(z) & (G = C_n) \\[.8em]
            \displaystyle 2 \det_{1 \le i, j \le n} d_z^{2j-2} \psi_{G,I(i)}(z) & (G = D_n)
        \end{cases}
    \end{align}
\end{theorem}

\section{$q$-Mellin--Barnes integral}

Set $r,s,n \in \mathbb{Z}_{\ge 0}$, $\kappa \in \mathbb{Z}$, with $0 \le s \le r$ and $0 \le n \le r$, and $I$ an injective map, $I : \{1,\ldots,n\} \hookrightarrow \{1,\ldots,r \}$ as before.
For linear forms, we write $x^\alpha = \ee^{\alpha(\log x)}$, $x^\mathsf{w} = \ee^{\mathsf{w}(\log x)}$, etc.
Let $q \in \mathbb{C}^\times$ with $|q| < 1$.
Then, we define a $q$-deformed Mellin--Barnes SW integral of type $G$ as follows:
\begin{align}
    \Phi_{G,I}^{(\kappa)} \left(\begin{matrix} a_1,\ldots,a_r \\ b_1,\ldots,b_s \end{matrix};z\right) = \frac{1}{|W_G|} \int \prod_{\alpha \in R_G} (x^\alpha;q)_\infty \prod_{\mathsf{w} \in \operatorname{wt}(\text{def})} \frac{\prod_{\alpha=1}^s (b_\alpha/x^{\mathsf{w}};q)_\infty}{\prod_{\alpha=1}^r (a_\alpha/x^{\mathsf{w}};q)_\infty} \prod_{i=1}^n w(x_i) \frac{\dd{x}_i}{2\pi\ii} \, ,
\end{align}
where
\begin{align}
    w(x) = 
    \begin{cases}
        z^{\log_q x} q^{\frac{\kappa}{2} \log_q^2 x} & (A_{n-1}) \\
        (z^{\log_q x} + z^{-\log_q x}) q^{\frac{\kappa}{2} \log_q^2 x} & (B_n, C_n, D_n)
    \end{cases}
\end{align}
We assume as before that the parameters $\{a_\alpha\}_{\alpha = 1, \ldots, r}$, $\{b_\alpha\}_{\alpha = 1, \ldots, s}$ are generic so that none of the poles of the shifted $q$-factorials are canceled by each other. 
Let $\mathcal{C}^\pm_\alpha$ be a contour enclosing
the poles of $(a_\alpha/x^{\pm};q)_\infty^{-1}$ located at $\left(a_\alpha q^m\right)^{\pm}$, $m \in \mathbb{Z}_{\ge 0}$ and define $\mathcal{C}_I = \mathcal{C}^+_{I(1)} \times \cdots \times \mathcal{C}^+_{I(n)}$ as before.
More concretely, the integral $\Phi_{G,I}$ is understood as follows: 
\begin{subequations}
\begin{itemize}
    \item $G = A_{n-1}$ :
\begin{align}
    \Phi_{G,I}^{(\kappa)} = \int_{\mathcal{C}_I} \theta(t x_1 \cdots x_n) \prod_{1 \le i \neq j \le n} \left(\frac{x_i}{x_j};q\right)_\infty \prod_{i=1}^n \frac{\prod_{\alpha=1}^s (b_\alpha/x_i;q)_\infty}{\prod_{\alpha=1}^r (a_\alpha/x_i;q)_\infty} z^{\log_q x_i} q^{\frac{\kappa}{2} \log_q^2 x_i}\frac{\dd{x}_i}{2\pi \ii x_i}
\end{align}
    \item $G = B_n, C_n, D_n$ :
    \begin{align}
    \Phi_{G,I}^{(\kappa)} = \frac{2^n n!}{|W_G|} \int_{\mathcal{C}_I} \prod_{\alpha \in R_G} (x^\alpha;q)_\infty \prod_{\mathsf{w}\in \operatorname{wt}(\text{def})} \frac{\prod_{\alpha=1}^s (b_\alpha/x^{\mathsf{w}};q)_\infty}{\prod_{\alpha=1}^r (a_\alpha/x^{\mathsf{w}};q)_\infty} \prod_{i=1}^n z^{\log_q x_i} q^{\frac{\kappa}{2} \log_q^2 x_i} \frac{\dd{x}_i}{2\pi \ii x_i}
    \end{align}
\end{itemize}
\end{subequations}
For a technical reason, we insert additional factors $\theta(t x_1 \cdots x_n)$ for the type $A_{n-1}$ integral with auxiliary parameter $t \in \mathbb{C}^\times$.

\subsection{Type $A_{n-1}$}

As a warm-up, we start with a single-variable integral, 
\begin{align}
    \varphi^{(\kappa)}_\alpha(z) = \oint_{\mathcal{C}^+_\alpha} \frac{\prod_{\beta=1}^s (b_\beta/x;q)_\infty}{\prod_{\beta=1}^r(a_\beta/x;q)_\infty} z^{\log_{q} x} q^{\frac{\kappa}{2} \log_q^2 x} \frac{\dd{x}}{2\pi\ii x} \, , \qquad \alpha = 1, \ldots, r \, .
\end{align}
For $\kappa = 0$, this solves the following $q$-shift equation,
\begin{align}
    \left[ \prod_{\alpha=1}^r \left( 1 - \frac{a_\alpha}{q^{d_z}} \right) - z \prod_{\alpha=1}^s \left( 1 - \frac{b_\alpha}{q^{d_z+1}} \right) \right] y(z) = 0 \, .
\end{align}
with $q^{d_z} f(x) = f(qx)$.
\begin{proposition}\label{prop:phi_single}
    Let $A = \prod_{\alpha=1}^r a_\alpha$, $B = \prod_{\alpha=1}^s b_\alpha$, and $\kappa \ge s - r$.
    We write $0_m = \{\underbrace{0,\ldots,0}_m\}$.
    Then, we have
    \begin{align}
        \varphi_\alpha^{(\kappa)}(z) 
        & = 
        \frac{\prod_{\beta=1}^s (b_\beta/a_\alpha;q)_\infty}{(q;q)_\infty\prod_{\beta(\neq\alpha)}^r(a_\beta/a_\alpha;q)_\infty} z^{\log_{q} a_\alpha} q^{\frac{\kappa}{2} \log_q^2 a_\alpha} \nonumber \\
        & \qquad \times
        \begin{cases}
        \displaystyle
         {}_s \phi_{\kappa+r-1} \left( 
        \begin{matrix}
            \{q a_\alpha / b_\beta \}_{\beta=1}^s \\ \{ q a_\alpha / a_\beta\}_{\beta(\neq\alpha)}^r, 0_\kappa
        \end{matrix} ; q, (-1)^\kappa \frac{B}{A} a_\alpha^{\kappa+r-s} q^{\frac{\kappa}{2}+r-s} z \right)
        & (\kappa \ge 0) \\
        \displaystyle
        {}_{|\kappa|+s} \phi_{r-1} \left( 
        \begin{matrix}
            \{q a_\alpha / b_\beta \}_{\beta=1}^s, 0_{|\kappa|} \\ \{ q a_\alpha / a_\beta\}_{\beta(\neq\alpha)}^r
        \end{matrix} ; q,(-1)^\kappa \frac{B}{A} a_\alpha^{\kappa+r-s} q^{\frac{\kappa}{2}+r-s} z \right)
        & (s - r \le \kappa < 0)
        \end{cases}
    \end{align}
\end{proposition}
\begin{proof}
    We first consider the case $\kappa = 0$.
    Recall that
    \begin{align}
        \frac{(z/q^{m};q)_\infty}{(z;q)_\infty} = (1-z/q^m) \cdots (1-z/q) = (-zq)^m q^{-{m \choose 2}} (q/z;q)_m \, .
    \end{align}
    Then, evaluating the residues of the poles, we have
    \begin{align}
        \varphi^{(\kappa=0)}_\alpha(z) 
        & = \sum_{m = 0}^\infty \frac{z^{\log_q a_\alpha + m}}{(q^{-m};q)_m (q;q)_\infty} \frac{\prod_{\beta=1}^s (q^{-m} b_\beta / a_\alpha;q)_\infty}{\prod_{\beta(\neq\alpha)}^s (q^{-m} a_\beta/a_\alpha;q)_\infty} \nonumber \\
        & = \frac{\prod_{\beta=1}^s (b_\beta/a_\alpha;q)_\infty}{(q;q)_\infty\prod_{\beta(\neq\alpha)}^r(a_\beta/a_\alpha;q)_\infty} z^{\log_{q} a_\alpha} \sum_{m=0}^\infty \frac{z^m}{(q^{-m};q)_m} \prod_{\beta=1}^s \frac{(q^{-m} b_\beta / a_\alpha;q)_\infty}{(b_\beta / a_\alpha;q)_\infty} \prod_{\beta(\neq \alpha)}^r \frac{(a_\beta/a_\alpha;q)_\infty}{(q^{-m} a_\beta/a_\alpha;q)_\infty} \nonumber \\
        & = \frac{\prod_{\beta=1}^s (b_\beta/a_\alpha;q)_\infty}{(q;q)_\infty\prod_{\beta(\neq\alpha)}^r(a_\beta/a_\alpha;q)_\infty} \nonumber \\
        & \quad \times z^{\log_{q} a_\alpha} \sum_{m=0}^\infty \left( (-1)^m q^{m \choose 2} \right)^{r-s} \left(\frac{B}{A} (a_\alpha q)^{r-s} z\right)^m \frac{\prod_{\beta=1}^s (q a_\alpha/b_\beta;q)_m}{(q;q)_\infty \prod_{\beta(\neq\alpha)}^r (qa_\alpha/a_\beta;q)_m} \, .
    \end{align}
    The last line agrees with the basic hypergeometric series \eqref{eq:q-hypergeom}.
    For the case $\kappa \neq 0$, we insert the following,
    \begin{align}
        q^{\frac{\kappa}{2} (\log_q a_\alpha + m)^2} & = q^{\frac{\kappa}{2} (\log_q^2 a_\alpha + m^2) + m \kappa \log_q a_\alpha} = q^{\frac{\kappa}{2} \log_q^2 a_\alpha + \kappa {m \choose 2}} (a_\alpha q^{\frac{1}{2}})^{m \kappa} \, ,
    \end{align}
    from which we obtain the result.
\end{proof}
\begin{remark}
    In general, we may insert an additional term, $\theta(t x)$, in the integral:    
    \begin{align}
        \oint_{\mathcal{C}^+_\alpha} \frac{\prod_{\beta=1}^s (b_\beta/x;q)_\infty}{\prod_{\beta=1}^r(a_\beta/x;q)_\infty} z^{\log_{q} x} q^{\frac{\kappa}{2}\log^2_q x} \theta(t x) \frac{\dd{x}}{2\pi\ii x} \, .
    \end{align}
    In this case, the integral is still written as a basic hypergeometric series with additional factors,
    \begin{align}
        \theta(t a_\alpha q^m) & = (-ta_\alpha)^{-m} q^{-{m \choose 2}} \theta(ta_\alpha) \, .
    \end{align}
\end{remark}

For $\kappa - n \ge s - r$, we define
\begin{align}
    \varphi^{(\kappa)}_{A_{n-1},\alpha}(z) & = q^{\frac{n}{2}\log_q^2 a_\alpha} \varphi_{\alpha}^{(\kappa-n)}(-z/t) \, .
    \label{eq:phi_A}
\end{align}
Then, we have the following formula.
\begin{theorem}\label{thm:q-MB_A}
    We write $a_I = \{a_{I(1)},\ldots,a_{I(n)}\}$ and $A_I = \prod_{i=1}^{n} a_{I(i)}$.  
    The $q$-deformed Mellin--Barnes SW integral of type $A_{n-1}$ is given by a $q$-Casoratian,
    \begin{align}
        \Phi^{(\kappa)}_{A_{n-1},I}(z) = \theta(t A_I;q) W_{A_{n-1}}(a_I) \det_{1 \le i, j \le n} \varphi^{(\kappa)}_{A_{n-1},I(i)}(zq^{-j+1}) \, .
    \end{align}
\end{theorem}
\begin{proof}
    The strategy is the same as in the previous case (Theorem~\ref{thm:MBSW_A}). 
    Let $|z_i| < 1$ for any $i \in \{1,\ldots,n\}$ and we write $\theta(\cdot) = \theta(\cdot;q)$.
    We consider a generalized integral,
    \begin{align}
        \tilde{\Phi}^{(\kappa)}_I & = \int_{\mathcal{C}_I} \theta(t x_1 \cdots x_n) \prod_{1 \le i \neq j \le n} \left(x_i/x_j;q\right)_\infty \prod_{i=1}^n \frac{\prod_{\alpha=1}^s (b_\alpha/x_i;q)_\infty}{\prod_{\alpha=1}^r (a_\alpha/x_i;q)_\infty} z_i^{\log_q x_i} q^{\frac{\kappa}{2} \log_q^2 x_i} \frac{\dd{x}_i}{2\pi \ii x_i} \nonumber \\
        & = \int_{\mathcal{C}_I} \theta(t x_1 \cdots x_n) \prod_{1 \le i < j \le n} \left( x_j^{-1} - x_i^{-1} \right) x_j \theta\left(\frac{x_i}{x_j}\right) \prod_{i=1}^n \frac{\prod_{\alpha=1}^s (b_\alpha/x_i;q)_\infty}{\prod_{\alpha=1}^r (a_\alpha/x_i;q)_\infty} z_i^{\log_q x_i} q^{\frac{\kappa}{2} \log_q^2 x_i} \frac{\dd{x}_i}{2\pi \ii x_i} \, ,
    \end{align}
    which is evaluated by summing up the residues,
    \begin{align}
        \tilde{\Phi}^{(\kappa)}_I & = \sum_{0 \le m_1,\ldots,m_n \le \infty} \theta(t A_I q^{\sum_{i=1}^n m_i}) \prod_{1 \le i < j \le n} \left( a_{I(j)}^{-1} q^{-m_j} - a_{I(i)}^{-1} q^{-m_i} \right) a_{I(j)} q^{m_j} \theta \left(\frac{a_{I(i)}}{a_{I(j)}} q^{m_i-m_j}\right) \nonumber \\
        & \qquad \times \prod_{i=1}^n \frac{z_i^{\log_{q} a_{I(i)} + m_i} q^{\frac{\kappa}{2} (\log_q a_{I(i)} + m)^2}}{(q^{-m_i};q)_{m_i}(q;q)_\infty} \frac{\prod_{\beta=1}^s(q^{-m_i} b_\beta/a_{I(i)};q)_\infty}{\prod_{\beta(\neq I(i))}^r (q^{-m_i}a_\beta/a_{I(i)};q)_\infty} 
    \end{align}
    By the shift relation of the theta function~\eqref{eq:theta_shift}, we have
    \begin{align}
        & \theta(t A_I q^{\sum_{i=1}^n m_i}) \prod_{1 \le i < j \le n} a_{I(j)} q^{m_j} \theta \left(\frac{a_{I(i)}}{a_{I(j)}} q^{m_i-m_j}\right) \nonumber \\
        & = (-t A_{I})^{-\sum_{i=1}^n m_i} q^{-{\sum_{i=1}^n m_i \choose 2}} \left(\prod_{1 \le i < j \le n} q^{m_j} \left( - \frac{a_{I(i)}}{a_{I(j)}} \right)^{-m_i+m_j} q^{-{m_i-m_j \choose 2}} \right) \theta(t A_I) W_{A_{n-1}}(a_I) \nonumber \\
        & = (-t A_{I})^{-\sum_{i=1}^n m_i} \left( \prod_{1 = 1}^n a_{I(i)}^{-nm_i + \sum_{j=1}^n m_j} \right) q^{-{\sum_{i=1}^n m_i \choose 2} - (n-1) \sum_{i=1}^n {m_i \choose 2} + \sum_{1 \le i < j \le n} m_i m_j} \theta(t A_I) W_{A_{n-1}}(a_I) \nonumber \\
        & = \left(\prod_{i=1}^n  \left(- t^{-1} a_{I(i)}^{-n} \right)^{m_i} q^{- n {m_i \choose 2}}  \right) \theta(t A_I) W_{A_{n-1}}(a_I) \, .
    \end{align}
    Then, we have
    \begin{align}
        \tilde{\Phi}^{(\kappa)}_I & = \theta(t A_I) W_{A_{n-1}} \prod_{i=1}^n z_i^{\log_q a_{I(i)}} q^{\frac{\kappa}{2} \log_q^2 a_{I(i)}} \sum_{0 \le m_1,\ldots,m_n \le \infty} \prod_{1 \le i < j \le n} \left( a_{I(j)}^{-1} q^{-m_j} - a_{I(i)}^{-1} q^{-m_i} \right) \nonumber \\
        & \qquad \times \prod_{i=1}^n \frac{(-t^{-1} a_{I(i)}^{\kappa-n} q^{\frac{\kappa}{2}} z_i)^{m_i} q^{(\kappa-n){m_i \choose 2}}}{(q^{-m_i};q)_{m_i}(q;q)_\infty} \frac{\prod_{\beta=1}^s(q^{-m_i} b_\beta/a_{I(i)};q)_\infty}{\prod_{\beta(\neq I(i))}^r (q^{-m_i}a_\beta/a_{I(i)};q)_\infty} \, ,
    \end{align}
    which is summarized into a determinant, 
    \begin{align}
        \tilde{\Phi}^{(\kappa)}_I & = \theta(t A_I) W_{A_{n-1}} \det_{1 \le i, j \le n} \varphi^{(\kappa)}_{A_{n-1},I(i)}(z_iq^{-j+1}) \, .
    \end{align}
    We obtain the result by taking the limit $z_i \to z$ for all $i\in\{1,\ldots,n\}$, $\lim_{z_i \to z} \tilde{\Phi}^{(\kappa)}_I = \Phi^{(\kappa)}_I$.
\end{proof}

\subsection{Type $B_n$, $C_n$, $D_n$}

We define a single-variable integral,
\begin{align}
    \varphi^{(\kappa)\pm}_\alpha(z) = \int_{\mathcal{C}_\alpha^+} \frac{\prod_{\beta = 1}^s (b_\beta/x;q)_\infty (b_\beta x;q)_\infty}{\prod_{\beta = 1}^r (a_\beta/x;q)_\infty (a_\beta x;q)_\infty} z^{\log_q x} q^{\frac{\kappa}{2} \log_q^2 x} \frac{\dd{x}}{2\pi\ii x} \, , \qquad \alpha = 1, \ldots, r \, ,
\end{align}
which will be a building block for the $q$-Mellin--Barnes integral of type $B_n, C_n, D_n$.
\begin{proposition}
    For $\kappa \ge s-r$, we have
    \begin{align}
        \varphi^{(\kappa)\pm}_\alpha(z) & = \frac{\prod_{\beta=1}^s (b_\beta/a_\alpha;q)_\infty(b_\beta a_\alpha;q)_\infty}{(q;q)_\infty\prod_{\beta(\neq\alpha)}^r(a_\beta/a_\alpha;q)_\infty\prod_{\beta=1}^r(a_\beta a_\alpha;q)_\infty} z^{\log_{q} a_\alpha} q^{\frac{\kappa}{2} \log_q^2 a_\alpha} \nonumber \\
        & \qquad \times {}_{r+s} \phi_{\kappa+2r-1} \left( 
        \begin{matrix}
            \{a_\alpha a_\beta\}_{\beta=1}^r, \{q a_\alpha / b_\beta \}_{\beta=1}^s, \\ \{ q a_\alpha / a_\beta\}_{\beta(\neq\alpha)}^r, \{ a_\alpha b_\beta \}_{\beta=1}^s, 0_{\kappa+r-s}
        \end{matrix} ; q,(-1)^\kappa \frac{B}{A} a_\alpha^{\kappa+r-s} q^{\frac{\kappa}{2}+r-s} z \right) \, .
    \end{align}
\end{proposition}
\begin{proof}
    The same proof is applied as in the previous case (Proposition~\ref{prop:phi_single}).
\end{proof}
\begin{align}
    \varphi^{(\kappa)}_{G,\alpha}(z) = 
    \begin{cases}
        \displaystyle
        \frac{\prod_{\beta=1}^s (b_\beta;q)_\infty}{\prod_{\beta=1}^r (a_\beta;q)_\infty} a_\alpha^{-\frac{1}{2}} \varphi_\alpha^{(\kappa-2n+1)\pm}(z) & (G = B_n) \\
        \varphi^{(\kappa-2n-2)\pm}_\alpha(z) & (G = C_n) \\
        \varphi^{(\kappa-2n+2)\pm}_\alpha(z) & (G = D_n) \\
    \end{cases}
    \label{eq:phi_BCD}
\end{align}
Then, we have the following.
\begin{theorem}\label{thm:q-MB_BCD}
    The $q$-deformed Mellin--Barnes SW integral of type $B_n$, $C_n$, $D_n$ is given by a (skew-)symmetrized $q$-Casoratian,
    \begin{align}
        \Phi^{(\kappa)}_{G,I}(z) = W_G(a_I) \times 
        \begin{cases}
            \displaystyle
            \det_{1 \le i, j \le n} \left( \varphi^{(\kappa)}_{G,I(i)}(z q^{n+\frac{1}{2}-j}) - \varphi^{(\kappa)}_{G,I(i)}(z q^{-n-\frac{1}{2}+j}) \right) & (G = B_n) \\
            \displaystyle
            \det_{1 \le i, j \le n} \left( \varphi^{(\kappa)}_{G,I(i)}(z q^{n+1-j}) - \varphi^{(\kappa)}_{G,I(i)}(z q^{-n-1+j}) \right) & (G = C_n) \\
            \displaystyle
            \det_{1 \le i, j \le n} \left( \varphi^{(\kappa)}_{G,I(i)}(z q^{n-j}) + \varphi^{(\kappa)}_{G,I(i)}(z q^{-n+j}) \right) & (G = D_n) 
        \end{cases}
    \end{align}
\end{theorem}
\begin{proof}
    By the shift relation of the theta function~\eqref{eq:theta_shift}, we have
    \begin{subequations}
    \begin{align}
        & \prod_{1 \le i < j \le n} a_{I(i)}^{-1} q^{-m_i} \theta\left(\frac{a_{I(i)}}{a_{I(j)}} q^{m_i-m_j} \right) \theta \left( a_{I(i)} a_{I(j)} q^{m_i + m_j} \right) \nonumber \\
        & = \left( \prod_{1 \le i < j \le n} q^{-m_i} \left( - \frac{a_{I(i)}}{a_{I(j)}} \right)^{-m_i+m_j} \left( - a_{I(i)}a_{I(j)} \right)^{-m_i-m_j} q^{-{m_i - m_j \choose 2} - {m_i + m_j \choose 2}} \right) W_{D_{n}}(a_I) \nonumber \\
        & = \left( \prod_{i=1}^n \left( a_{I(i)} q^{\frac{1}{2}} \right)^{-(n-1)m_i} q^{-2(n-1) {m_i \choose 2}} \right) W_{D_n}(a_I) \, , 
    \end{align}
    \begin{align}
        a_{I(i)}^{-\frac{1}{2}} q^{-\frac{1}{2}m_i} \theta(a_{I(i)} q^{m_i}) & = \left( - a_{I(i)} q^{\frac{1}{2}} \right)^{-m_i} q^{-{m_i \choose 2}} a_{I(i)}^{-\frac{1}{2}} \theta(a_{I(i)}) \, , \\
        a_{I(i)}^{-1} q^{-m_i} \theta(a_{I(i)}^2 q^{2m_i}) & = \left( a_{I(i)} q^{\frac{1}{2}} \right)^{-4m_i} q^{-4{m_i \choose 2}} a_{I(i)}^{-1} \theta(a_{I(i)}^2) \, .
    \end{align}
    \end{subequations}
    Then, we may apply the determinantal expressions~\eqref{eq:q-SW_det}. 
    The remaining computation is in parallel with the case $G = A_{n-1}$.
\end{proof}
\begin{remark}
    The shift of $\kappa$ in $\varphi_G^{(\kappa)}$ shown in \eqref{eq:phi_A} and \eqref{eq:phi_BCD} agrees with the dual Coxeter number $h^\vee$ for $G = A_{n-1}$, $B_n$, $C_n$, while it is given by $2 h^\vee$ for $G = C_n$.
\end{remark}

\appendix

\section{Special functions}

\subsection*{$q$-shifted factorial}
We define the $q$-shifted factorial,
\begin{align}
    (z;q)_m = \prod_{n=0}^{m-1} (1 - z q^n) \, .
\end{align}
For $|q| < 1$, we have 
\begin{align}
    (z;q)_\infty = \prod_{n=0}^{\infty} (1 - z q^n) \, .
\end{align}

\subsection*{Theta function}
For $|q| < 1$, the theta function is defined by
\begin{align}
    \theta(z;q) = (z;q)_\infty (q/z;q)_\infty = \frac{1}{(q;q)_\infty} \sum_{n \in \mathbb{Z}} (-1)^n q^{n \choose 2} z^n \, , \label{eq:theta_def}
\end{align}
obeying the shift relation,
\begin{align}
    \theta(zq^m;q) = (-z)^{-m} q^{-{m \choose 2}} \theta(z;q) \, , \quad m \in \mathbb{Z} \, . \label{eq:theta_shift}
\end{align}

\subsection*{Hypergeometric functions}
For $|z| < 1$, we define the hypergeometric functions,
\begin{subequations}
\begin{align}
    {}_r F_s \left( \begin{matrix} a_1,\ldots,a_r \\ b_1,\ldots,b_s \end{matrix} ; z \right) & = \sum_{m=0}^\infty \frac{(a_1)_m \cdots (a_r)_m}{(b_1)_m \cdots (b_s)_m} \frac{z^m}{m!} \\
    {}_r \phi_s \left( \begin{matrix} a_1,\ldots,a_r \\ b_1,\ldots,b_s \end{matrix} ; q, z \right) & = \sum_{m=0}^\infty \left( (-1)^m q^{m \choose 2} \right)^{s-r+1} \frac{(a_1;q)_m \cdots (a_r;q)_m}{{(b_1;q)_m \cdots (b_s;q)_m}} z^m \label{eq:q-hypergeom}
\end{align}
\end{subequations}

\bibliographystyle{ytamsalpha}
\bibliography{ref}

\end{document}